\documentclass[conference,romanappendices,10pt]{IEEEtran}

\usepackage{tikz}
\usepackage[normalem]{ulem}
\usepackage{amsthm,amsfonts,amsmath,amssymb}
\usepackage{algorithm,algpseudocode}
\usepackage{arydshln}
\usepackage{mathtools}
\usepackage[hidelinks]{hyperref} 
\usepackage{cleveref}
\usepackage{ifthen}
\usepackage{wrapfig} 
\usepackage{xfrac}
\usepackage{bbm}
\usepackage{subcaption}
\usepackage{enumerate} 
\allowdisplaybreaks[2]
\usepackage{multicol,float}
\theoremstyle{plain}
\newtheorem{theorem}{Theorem}
\newtheorem{lemma}{Lemma}

\theoremstyle{remark}

\theoremstyle{definition}

\newif\ifnotes
\notestrue
\newcommand{\snote}[1]{\ifnotes{{#1}}\fi}
\newcommand{\christina}[1]{\ifnotes{{#1}}\fi}
\newcommand{\pavlos}[1]{\ifnotes{{#1}}\fi}

\renewcommand{\Pr}{\mbox{${\mathbb P}$} }

\newcommand{\testmatrix}{\mathbf{G}}

\newcommand{\znoiseProb}{\mathit{z}}

\newcommand{\Prob}{\mbox{${\mathbb P}$} }

\newcommand{\graph}{\mathit{G}}
\newcommand{\vertexSet}{\mathcal{V}}

\newcommand{\edgeSet}{\mathcal{E}}
\newcommand{\edgeSubset}{\mathcal{I}}
\newcommand{\vertexIndex}{\mathit{v}}
\newcommand{\edgeIndex}{\mathit{e}}
\newcommand{\component}{\mathit{C}}
\newcommand{\disjointSetIdx}{\mathit{d}}
\newcommand{\disjointSetIdxSecond}{\mathit{b}}
\newcommand{\disjointSets}{\mathit{D_{\component}}}
\newcommand{\disjointOuterSets}{\mathit{D_{\component,out}}}
\newcommand{\disjointInnerSets}{\mathit{D_{\component,in}}}

\newcommand{\setDef}{\mathit{K}}

\newcommand{\nItems}{\mathit{n}}
\newcommand{\nFamilies}{\mathit{F}}
\newcommand{\familyIndex}{\edgeIndex}
\newcommand{\familyInfectRate}{\mathit{q}}
\newcommand{\nDef}{\mathit{k}}
\newcommand{\sparseRegimePar}{\mathit{\alpha}}
\newcommand{\nDefFamilies}{\mathit{k_f}}
\newcommand{\sparseRegimeFamilyPar}{\mathit{\alpha_f}}
\newcommand{\nMembersSymmetric}{\mathit{M}}
\newcommand{\nDefMembersSymmetric}{\mathit{k_m}}
\newcommand{\memberInfectRateSymmetric}{\mathit{p}}
\newcommand{\memberInfectProb}{\memberInfectRateSymmetric_\memberIndex}

\newcommand{\nDefMembers}{\mathit{k_m^\familyIndex}}
\newcommand{\nDefMembersOverlap}{\mathit{k_m^\edgeSubset}}
\newcommand{\nDefMembersDisjointSet}{k_m^\disjointSetIdx}
\newcommand{\memberInfectRate}{\mathit{p_\familyIndex}}

\newcommand{\estimInfectRateDisjoint}{\mathit{\hat{p}_\disjointSetIdx}}
\newcommand{\estimInfectRateDisjointSecond}{\mathit{\hat{p}_\disjointSetIdxSecond}}
\newcommand{\thresh}{\mathit{\theta}}

\newcommand{\sampleSet}{\mathit{r}}

\newcommand{\mixedSample}{\mathit{z}}

\newcommand{\createSample}{\mathit{SelectRepresentatives}}
\newcommand{\adapt}{\mathit{AdaptiveTest}}

\newcommand{\nTests}{\mathit{T}}
\newcommand{\testResult}{\mathit{Y}}
\newcommand{\testIndex}{\mathit{\tau}}

\newcommand{\memberIndex}{\vertexIndex}
\newcommand{\defVector}{\mathbf{U}}
\newcommand{\defVariable}{\mathit{U}}
\newcommand{\defectValue}{\mathit{u}}
\newcommand{\defFamilyVariable}{\mathit{X}}
\newcommand{\defFamilyVariables}{\mathbf{X}}
\newcommand{\defFamilyValues}{\mathbf{x}}
\newcommand{\defFamilyValuesDict}{\mathbf{\mathcal{X}}}

\newcommand{\binEntropy}{\mathit{h_2}}
\newcommand{\entropy}{\mathit{H}}
\newcommand{\Ber}{\text{Ber}}
\newcommand{\boundRegime}{38\%}

\newcommand{\nOnesPerRow}{\mathit{c}}
\newcommand{\nBlocks}{\mathit{b}}
\newcommand{\blockInfectRate}{\mathit{w}}

\newcommand\blfootnote[1]{%
	\begingroup
	\renewcommand\thefootnote{}\footnote{#1}%
	\addtocounter{footnote}{-1}%
	\endgroup
}

\allowdisplaybreaks

\begin{document}
\title{\bf Group testing for overlapping communities}
\author{ \IEEEauthorblockN{
Pavlos Nikolopoulos\IEEEauthorrefmark{1},
Sundara Rajan Srinivasavaradhan\IEEEauthorrefmark{2},
Tao Guo\IEEEauthorrefmark{2}, 
Christina Fragouli\IEEEauthorrefmark{2}, 
Suhas Diggavi\IEEEauthorrefmark{2}} \IEEEauthorblockA{\IEEEauthorrefmark{1}École polytechnique fédérale de Lausanne (EPFL), Switzerland} \IEEEauthorblockA{\IEEEauthorrefmark{2}University of California Los Angeles (UCLA), USA}}\maketitle

\begin{abstract}
In this paper, we propose algorithms that leverage a known community structure to make group testing more efficient.
We consider a population organized in connected communities: each individual participates in one or more communities, and the infection probability of each individual depends on the communities (s)he participates in. Use cases include students who  participate in several classes, and workers who share common spaces.
Group testing reduces the number of tests needed to identify the infected individuals by pooling diagnostic samples and testing them together.
We show that making testing algorithms aware of the community structure,  can significantly reduce the number of tests needed both for adaptive and non-adaptive group testing. \blfootnote{This work was supported in part by NSF grants \# 2007714, 1705077 and UC-NL grant LFR-18-548554.}
\end{abstract}


\section{Introduction}

A significant challenge to re-open society in the
presence of a pandemic is the need for large scale, reliable testing. Group testing is a strategy that can help reduce the number of tests and increase reliability by pooling together diagnostic samples. Accordingly, it is attracting significant attention: several countries  (India, Germany, US, China) are currently deploying group testing strategies to help them quickly and reliably identify infected people \cite{GroupTest-implement1,GroupTest-implement2-FDA}.

In this paper, we build group testing algorithms around an idea ``whose time has come": we propose to leverage a known community structure to make group testing more efficient. Although traditionally the work in group testing assumes ``independent" infections, we note that today it is totally feasible to keep track of community structure - several apps are already doing so \cite{TrackingCommunity-1,TrackingCommunity-Google}. Moreover, our approach is well aligned with the need for  independent grassroots testing  (schools testing their students, companies their workers) where the community structure is explicit (shared classrooms, shared common spaces). 


We find that taking into account the community structure
can reduce the number of tests we need significantly below the well known combinatorial bound \cite{GroupTestingMonograph}, the best we can hope for when not taking this structure into account.
Moreover, it enlarges the regime where group testing can offer benefits over individual testing. Indeed,
a limitation of group testing is that it does not offer benefits when \pavlos{the number of infected people} grows linearly with the size of the population considered~\cite{LinBndIndv2,LinBndIndv-3-Ungar-1960}. 
Taking into account the community structure allows to identify and remove from the population large groups of infected members, thus reducing their proportion and converting  a linear to a sparse regime identification. 
However, we also find that, although our algorithms do not require exact knowledge of the infection probabilities, they do need to correctly know the community structure, and in particular, the community overlaps:  not taking into account the overlaps (assuming communities are disconnected) can deteriorate the performance.
Our main contributions include:\\
$\bullet$ We derive a lower bound on the number of tests, that generalizes the counting bound~\cite{GroupTestingMonograph} to overlapping communities.\\ 
$\bullet$ We propose a new adaptive algorithm that requires fewer tests than traditional adaptive algorithms to recover the infection status of all individuals without error.\\
$\bullet$ We propose two nonadaptive algorithms that leverage the community structure to improve reliability over traditional nonadaptive testing. 
One leverages the structure at the encoder side with a novel test design, 
while the other one accounts for it at the decoder side with a new decoding algorithm that is based on loopy belief propagation (LBP) and is generic enough to work on any structure.

The paper is organized as follows: 
we give known results in~\Cref{sec:back}, 
our model in~\Cref{sec:notation}, 
the lower bound in \Cref{sec:lower-bounds}, and non-adaptive algorithms  in \Cref{section-algorithm,sec:LBP}. 
Our numerical evaluation is in \Cref{section-experiments}.


\section{Background}
\label{sec:back}
\pavlos{
In group testing, a test $\testIndex$ takes as input samples from $n_\testIndex$ individuals, pools them together and  outputs a single value:  positive if any one of the samples is infected, and negative if no one is.  
More precisely, let ${\defVariable_i=1}$  when individual $i$ is infected and $0$ otherwise. 
$\testResult_\testIndex$ takes a binary value calculated as $\testResult_\testIndex= \bigvee_{i\in \delta_{\testIndex}} \defVariable_i$, 
where $\bigvee$ stands for the \texttt{OR} operator (disjunction) and $\delta_{\testIndex}$ is the group of people participating in the test.

Traditional group testing assumes
a population of $\nItems$ members out of which some are infected. 
Two infection models are considered: (i) in the {\em combinatorial model}, there is a fixed number of infected members $\nDef$, selected uniformly at random among all sets of size $\nDef$; (ii) in the {\em probabilistic model}, each item is infected independently of all others with probability $\memberInfectRateSymmetric$, so that the expected number of infected members is $\bar{\nDef} = \nItems\memberInfectRateSymmetric$. 

Several more detailed setups and various testing algorithms have been explored in the literature. These include:

\noindent$\bullet$
{\em Sparse vs. linear scaling regime}: assume $\nDef=\Theta(\nItems^\alpha)$, we say we operate in the linear regime if $\alpha=1$; in the sparse regime if $0\leq \alpha < 1$; in the very sparse regime if $k$ is constant.

\noindent$\bullet$
{\em Adaptive vs. non-adaptive testing}: In adaptive testing, we use the outcome of previous tests to decide what tests to perform next. An example of adaptive testing is {\em binary splitting}, which implements a form of binary search and is known to be optimal when the number of infected members is unknown.
Non-adaptive testing constructs, in advance, a \emph{test matrix} $\testmatrix\in \{0,1\}^{\nTests\times \nItems}$  where each row corresponds to one test, each column to one member, and the non-zero elements in each row determine the set $\delta_{\tau}$.  
Although adaptive testing generally requires fewer tests than non-adaptive, non-adaptive testing is more practical \christina{when} all tests can be executed in parallel.


Inhere, we reuse the following well-established results (see \cite{GroupTestingMonograph,CntBnd,GroupTestingBook} and references therein):

\noindent$\bullet$
In the combinatorial model, since $\nTests$ tests allow to distinguish among $2^\nTests$ combinations of test outputs, 
we need 
$\nTests \geq \log_2{\binom{\nItems}{\nDef}}$ to identify $\nDef$ randomly infected people out of $\nItems$.
This is known as the {\bf counting  bound}~\cite{GroupTestingMonograph,CntBnd,GroupTestingBook} and implies that we cannot use less than  $\nTests=\mathcal{O}(\nDef\log \frac{\nItems}{\nDef})$ tests. 
In the probabilistic model, a similar bound has been derived for the number of tests needed on average: $\nTests \geq \nItems \binEntropy\left(\memberInfectRateSymmetric\right)$,
where $\binEntropy$ is the binary entropy function.
}

\noindent$\bullet$ 
Noiseless adaptive testing allows to achieve the counting bound for $\nDef=\Theta(\nItems^\alpha)$ and $0\leq \alpha < 1$; 
for non-adaptive testing, this is also true of $0\leq \alpha <  0.409$, 
if we allow a vanishing (with $\nItems$) error~\cite{GroupTestingMonograph,coja-oghlan19,coja-oghlan20a}.

\noindent$\bullet$ 
In the linear regime ($\alpha=1$), group testing offers little benefits over individual testing.
In particular, if the infection rate $\sfrac{\nDef}{\nItems}$ is more than $\boundRegime$,
group testing does not use fewer tests than one-to-one (individual) testing
unless high identification-error rates are acceptable~\cite{LinBndIndv2,LinBndIndv-3-Ungar-1960,LinBndIndv1}.

\section{Model and Notation}
\label{sec:notation}
\subsection{\pavlos{Community model}}
\label{sec:notation:models}

Our work extends the above results by assuming an overlapping community structure: 
members may belong to one or more communities---hence they are infected according to new combinatorial and probabilistic models 
that are slightly different from the traditional ones and depend on how the communities overlap (\Cref{sec:infect-model}). 
Given these new infection models, new lower bounds can be derived (\Cref{sec:lower-bounds})
for both the combinatorial and probabilistic case. 
Our analysis shows that these bounds can be significantly lower than the above-mentioned counting bounds. 

More formally, we assume that all members of the entire population $\vertexSet = \{1,2,\cdots,\nItems\}$ are organized in a known structure,
which can be perceived in the form of a hypergraph $\graph(\vertexSet,\edgeSet)$: 
each vertex $\vertexIndex \in \vertexSet$ 
corresponds to an individual, that we simply call a \textit{member}, 
and each edge $\edgeIndex \in \edgeSet$ 
indicates which members belong to the same \textit{community}. 
Since $\graph$ is a hypergraph, 
an edge may connect any number of vertices; 
hence, a member may belong to one or more communities.
The number of communities that a member belongs to is called the \textit{degree} of the member. 
There exist $\nFamilies$ communities in total.

The hypergraph $\graph$ may be decomposed into connected components, where
each component $\component (\vertexSet_\component, \edgeSet_{\component})$ is a sub-hypergraph.
For component $\component$ let $\disjointSets$ 
contain the nonempty subsets of 
the standard partition of the hyperedges in $\edgeSet_\component$; 
in the example of Fig.~\ref{fig:outersets}, $\edgeSet_\component$ consists of $3$ hyperedges, and  $\disjointSets$ contains $7$ disjoint sets.
\begin{wrapfigure}{r}{0.45\linewidth}	
	\vspace{-10pt}
	\begin{center}	
		\includegraphics[scale=0.15]{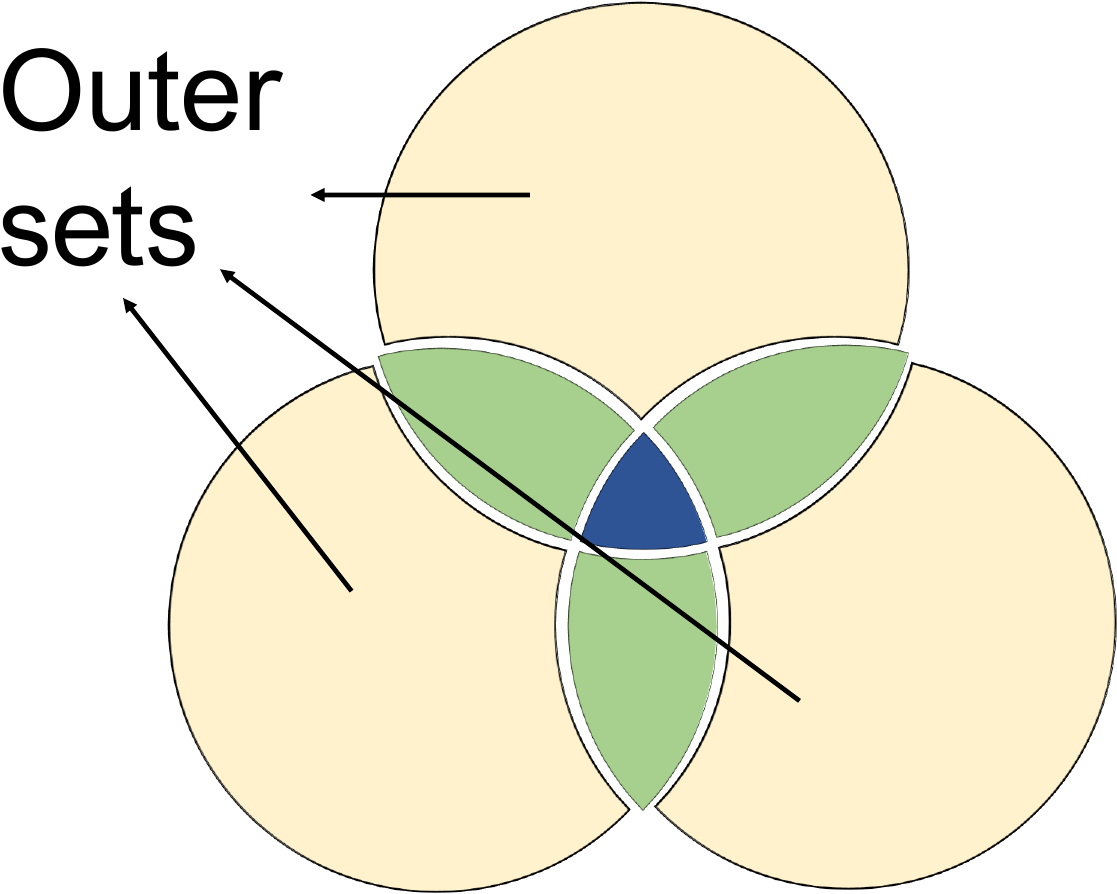}
		\caption{Standard partition.}
		\label{fig:outersets}
	\end{center}
	\vspace{-10pt}
\end{wrapfigure}
For each set $\disjointSetIdx \in \disjointSets$,  let $\vertexSet_\disjointSetIdx$ denote the set of members it contains.
Because of the partition, all members of $\vertexSet_\disjointSetIdx$ belong to the same community or set of communities, 
which we denote with $\edgeSet_{\vertexSet_\disjointSetIdx}$---hence, they all have the same degree.
As described in the next section, these members  get infected according to some common infection principle.
We distinguish 2 kinds of sets in $\disjointSets$:
(a) the ``outer'' sets: $\disjointOuterSets \triangleq \left\{\disjointSetIdx \in \disjointSets : \nexists \disjointSetIdxSecond \in \disjointSets \text{ s.t. } \edgeSet_{\vertexSet_\disjointSetIdxSecond} \subset \edgeSet_{\vertexSet_\disjointSetIdx} \right\}$, 
and (b) the ``inner'' sets: $\disjointInnerSets \triangleq \disjointSets \setminus \disjointOuterSets$.  
Fig.~\ref{fig:outersets} illustrates the 3 outer (yellow) and 4 inner (green, blue) sets.
Note that the members of inner sets have always a higher degree than those of the outer sets.

\subsection{Infection models}\label{sec:infect-model} 
We consider the following infection models, that parallel the ones in the traditional setup of~\Cref{sec:back}.\\
$\bullet$ {\bf Combinatorial Model (I).}  $\nDefFamilies$ of the communities have at least one infected member (we will call these {\em infected communities}). 
The rest of the communities have no infected members. 
Any combination of infected communities has the same chance of occurring. 
In each infected community, there are $\nDefMembers$
infected members, 
out of which $\nDefMembersOverlap$ 
are shared with a subset of infected communities $\edgeSubset \subseteq \edgeSet$.
The infected communities (resp. infected community members) are chosen uniformly at random out of all communities (resp. members that belong to the same communities).\\
$\bullet$ {\bf Probabilistic Model (II).} A community $\edgeIndex$ becomes infected i.i.d. with probability $\familyInfectRate$. 
If a member $\vertexIndex$ of an infected community $\edgeIndex$ belongs \textit{only} to that community (i.e. has degree $1$), then it becomes infected w.p.~$\memberInfectProb = \memberInfectRate$, independently from the other members (and other communities).
Also, if $\vertexIndex$ belongs to a subset of infected communities $\edgeSubset \subseteq \edgeSet$, 
it is considered to be infected by either of these communities.
So, given all the infection rates of these communities $\left\{\memberInfectRate: \edgeIndex \in \edgeSubset\right\}$,
we say that $\memberIndex$ becomes infected w.p.:
$\memberInfectProb
= 1 - \prod_{\familyIndex \in \edgeSubset}(1-\memberInfectRate)$.
Note that since each member gets infected by either of the infected communities it belongs to, the product of the RHS term above becomes smaller as $|\edgeSubset|$ increases.
Last, if $\vertexIndex$ does not belong to an infected community, then $\memberInfectProb = 0$.

We make two remarks: 
First, although the communities are infected independently, their structure causes a dependent infection model;
in fact, the way communities overlap determines the infection probability of their shared members. 
Second, our model captures situations where infection is determined by participation in a community rather than the status of community members.  
Albeit simplistic, we think that this model can be useful in real pandemics.
Since the exact community structure of the entire population of a country or a continent can never be known to the test designer, 
we expect that graph $\graph$ only partially describes the reality: 
there might be members that do not belong to $\vertexSet$, yet interact with them in unknown ways, 
or there might be communities that are simply not captured due to unknown member interactions.
In such a case, assuming that communities become infected independently seems a simple yet reasonable model to use. 
However, once a few communities in $\graph$ get infected, 
we expect that the infection probability of a member will increase with the number of infected communities it belongs to, 
which is captured by our model in the computation of $\memberInfectProb$.

\subsubsection*{Non-overlapping communities}
A special case of our community framework is when the communities have no overlap; this scenario is investigated in our prior work~\cite{GroupTesting-community,aistats-paper}, 
where algorithms that take into account the non-overlapping structure are explored.
In our experiments, however, these algorithms do not always perform well when communities overlap (e.g. see Fig.~\ref{fig:adaptive}); 
in fact, there may be cases where they perform worse than traditional group testing.
\pavlos{
Also, the idea of using side-information from a community structure in decoding of group tests has been identified in~\cite{zhu2020noisy,goenka2020contact}, independently from this work. 
That work is complementary to ours;
we focus more on test designs rather than decoding, for which we use well-known algorithms such as COMP and LBP.}

\section{Lower bound on the number of tests}
\label{sec:lower-bounds}
We  compute the minimum number of tests needed to identify all infected members under the zero-error criterion in both community models (I) and (II). All proofs can be found in the Appendix.
\begin{theorem}[Community bound for combinatorial model (I)]
	\label{thm:combinatorialBound} 
	Any algorithm that identifies all $\nDef$ infected members without error requires a number of tests $\nTests$ satisfying:
	\begin{equation}
	\nTests \ge \log_2{\binom{\nFamilies}{\nDefFamilies}} + \sum_{\component \in \graph} \sum_{\disjointSetIdx \in \disjointSets} \log_2{\binom{|\vertexSet_\disjointSetIdx|}{\nDefMembersDisjointSet}},
	\label{eq:combinatorialBound}
	\end{equation}
\end{theorem}
\noindent \pavlos{where $|\vertexSet_\disjointSetIdx|$ (resp. $\nDefMembersDisjointSet$) is the number of members (resp. infected members) of each disjoint set $\disjointSetIdx$ in $\disjointSets$}.

\subsubsection*{Observation} 
Consider a usual epidemic scenario, 
where the population is composed of a large number of communities with members that have close contacts (e.g. relatives, work colleagues, students who attend the same classes, etc.). 
In such a case, one should expect that most all members of each infect community are infected, 
even though the number of infected communities $\nDefFamilies$ and the overall number of infected members $\nDef$ may still follow a sparse regime 
(i.e., $\nDefFamilies = \Theta(\nFamilies^\sparseRegimeFamilyPar)$ and $\nDef = \Theta(\nItems^\sparseRegimePar)$ for $\sparseRegimeFamilyPar, \sparseRegimePar \in [0,1)$).
\Cref{thm:combinatorialBound} shows the significant benefit of taking the community structure into account in the test design:
the community bound increases almost \textit{linearly} with $\nDefFamilies$, 
as opposed to $\nDef$, which is what happens with the traditional counting bound (that does not account for any community structure).
This is due to the second term of~\Cref{eq:combinatorialBound} tending to $0$ and
$
\log_2 {\binom{\nFamilies}{\nDefFamilies}} \sim
\nDefFamilies \log_2 {\frac{\nFamilies}{\nDefFamilies}} \sim (1-\sparseRegimeFamilyPar) \nDefFamilies \log_2 {\nFamilies} \label{eq:asymptoticRelation1}
$.

\snote{\begin{theorem}[Community bound for probabilistic model (II)]
	\label{thm:probabilisticBound}
	Any algorithm that identifies all infected members without error requires a number of tests $\nTests$ satisfying:
	\begin{align}
	    &T \geq  \nFamilies \binEntropy(\familyInfectRate) + \sum_{\vertexIndex=1}^n \sum_{\edgeSubset \subseteq \edgeSet_\vertexIndex} \familyInfectRate^{|\edgeSubset|}(1-\familyInfectRate)^{|\edgeSet_\vertexIndex|-|\edgeSubset|}  \binEntropy(\prod_{\familyIndex \in \edgeSubset}(1-\memberInfectRate)) \nonumber\\
	    &- \sum_{\familyIndex=1}^\nFamilies (1-\familyInfectRate+\familyInfectRate(1-\memberInfectRate)^{|S_e|}) \binEntropy\left (\frac{1-\familyInfectRate}{1-\familyInfectRate+\familyInfectRate(1-\memberInfectRate)^{|S_e|}}\right ),
	    \label{eq:probabilisticBound}
	\end{align}
where $\edgeSet_\vertexIndex$ is the set of communities that member $\vertexIndex$ belongs to, $\edgeSubset$ is the subset of infected communities in $\edgeSet_\vertexIndex$, $S_e$ is the set of members who \textit{only} belong to community $\familyIndex$.
\end{theorem}

Theorem~\ref{thm:probabilisticBound} extends from zero-error recovery to constant-probability recovery,
if we apply Fano’s inequality in a similar way to Thm 1 of~\cite{prior}---in that case the right-hand side of \eqref{thm:probabilisticBound} is multiplied by the desired probability of correctly identifying all infected members.}

\section{Algorithms}\label{section-algorithm}
In this section, we provide group-testing algorithms for the noiseless case that leverage the community structure.
We start from adaptive algorithms
and then proceed to  non-adaptive.

\subsection{Adaptive algorithm}
\label{sec:adapt}

\begin{algorithm}[b!]
\caption{Adaptive Community Testing ($\graph\left(\vertexSet,\edgeSet\right)$)}
\label{algorithm-noiseless} 
$\hat{\defVariable}_\memberIndex$ is the estimated infection state of member $\memberIndex$ (``+'' or ``-'').\\
$\hat{\defVariable}_{\mixedSample}$ is the estimated infection state of a mixed sample $\mixedSample$.
\begin{algorithmic}[1]
	\For {$\disjointSetIdx \in \disjointOuterSets$, $\forall \component \in \graph$}
		\State{$\sampleSet_\disjointSetIdx \leftarrow \createSample\left( \vertexSet_\disjointSetIdx\right)$}
	\EndFor
	\State{$\left\{ \hat{\defVariable}_{\mixedSample(\sampleSet_\disjointSetIdx)}  \right\}   \leftarrow \adapt \left( \left\{\mixedSample(\sampleSet_\disjointSetIdx \right\} \right) $}
	\State{Set $A:=\emptyset$}
	\For {$\component \in \graph$}
		\For {$\disjointSetIdx \in \disjointOuterSets$}
			\If {$\hat{\defVariable}_{\mixedSample(\sampleSet_\disjointSetIdx)} = $ ``positive''}
			\State{Individually test $\vertexSet_\disjointSetIdx$: $\hat{\defVariable}_\memberIndex \leftarrow \defVariable_\memberIndex$, $\forall \memberIndex \in \vertexSet_\disjointSetIdx$.}
			\State{$\estimInfectRateDisjoint \leftarrow \sfrac{1}{|\vertexSet_\disjointSetIdx|} \cdot \sum_{\vertexIndex \in \vertexSet_\disjointSetIdx} \bf{1}_{ \left\{\hat{\defVariable}_\memberIndex = \text{ 'positive'} \right\} }$}
			\Else
			\State{$A \leftarrow A \cup \left\{\memberIndex : \memberIndex \in \vertexSet_\disjointSetIdx \right\}$}
			\EndIf
		\EndFor
		\For {$\disjointSetIdxSecond \in \disjointInnerSets$ (in increasing order of degree)} 
		\If 
		{$\exists \disjointSetIdx \in \disjointSets$ s.t.  $\edgeSet_{\vertexSet_\disjointSetIdx} \subset \edgeSet_{\vertexSet_\disjointSetIdxSecond}$ \& $\estimInfectRateDisjoint > \thresh$} 
		\State{Individually test $\vertexSet_\disjointSetIdxSecond$: $\hat{\defVariable}_\memberIndex \leftarrow \defVariable_\memberIndex$, $\forall \memberIndex \in \vertexSet_\disjointSetIdxSecond$.}
		\State{$\estimInfectRateDisjointSecond \leftarrow \sfrac{1}{|\vertexSet_\disjointSetIdxSecond|} \cdot \sum_{\vertexIndex \in \vertexSet_\disjointSetIdxSecond} \bf{1}_{ \left\{\hat{\defVariable}_\memberIndex = \text{ 'positive'} \right\} }$}
		\Else
		\State{$A \leftarrow A \cup \left\{\memberIndex : \memberIndex \in \vertexSet_\disjointSetIdxSecond \right\}$}
		\EndIf
		\EndFor
	\EndFor
	\State{$\left\{\hat{\defVariable}_{\memberIndex}: \memberIndex \in A \right\}  = \adapt \left(A\right)$}\\
	\Return $\left[\hat{\defVariable}_1,\ldots,\hat{\defVariable}_\nItems\right]$
\end{algorithmic}
\end{algorithm}

\Cref{algorithm-noiseless} describes our adaptive algorithm. 
It is built on top of traditional adaptive testing,which we will generally denote as $\adapt()$.
$\adapt()$ is an abstraction of any existing (or future) adaptive algorithm that assumes independent infections. 
We distinguish $2$ different inputs for $\adapt()$:
(a) a set of members; or
(b) a set of \textit{mixed samples}. 
A mixed sample is created by pooling together samples from multiple members. 
For example, mixed sample $\mixedSample(\sampleSet_\disjointSetIdx)$ is an pooled sample of some representative members $\sampleSet_\disjointSetIdx$ from disjoint set $\disjointSetIdx$. 
Because we only care whether a mixed sample is positive or not, 
we can treat it in the same way as an individual sample---hence use group testing to identify the state of mixed samples as we do for individuals.

\noindent\textbf{Part 1:}
For each component of the graph $\graph$, we first identify the outer sets $\disjointOuterSets$.
Then, from each outer set $\disjointSetIdx$, we select a representative subset of members $\sampleSet_\disjointSetIdx$, whose samples are pooled together into a mixed sample $\mixedSample(\sampleSet_\disjointSetIdx)$. 
There can be many selection methods for $\createSample()$; however, we typically use uniform (random) sampling without replacement.
Finally, we determine the state of all mixed samples (line 4). 

\noindent\textbf{Part 2:}
We treat $\hat{\defVariable}_{\mixedSample(\sampleSet_\disjointSetIdx)}$ as a rough estimate of the infection regime inside each set $\disjointSetIdx$:
if $\hat{\defVariable}_{\mixedSample(\sampleSet_\disjointSetIdx)}$ is positive, we consider $\disjointSetIdx$ to be heavily infected and we individually test its members (line 9);
otherwise, we consider it lightly infected and we include its members in set $A$ (line 12).  
For our rough estimate of the infection regime to be a good one, we choose the number of representatives based on some prior information about infection rate of each outer set; for example if $\memberInfectRate < \boundRegime$ then only one representative is enough, 
otherwise pooling together the entire set is one's best option. 
Note that the exact knowledge of $\memberInfectRate$ and a rough prior may be easily acquired. For example, in realistic scenarios, where the infection rates are not expected to be very low inside the communities, 
pooling together the entire outer set is a good heuristic. 

Due to individually testing the heavily-infected outer sets, 
we obtain more accurate estimates of their infection rates, $\estimInfectRateDisjoint$, by computing the average proportion of infected members (line 10).
We use these estimates to decide how to test the inner sets of the component:
if an outer set $\disjointSetIdx$ exists whose members belong to a subset of communities in $\edgeSet_{\vertexSet_\disjointSetIdxSecond}$ and its estimated infection rate $\estimInfectRateDisjoint$ is above a threshold $\thresh$, 
then members of $\vertexSet_\disjointSetIdxSecond$ are tested individually (line 17) and a new estimate $\estimInfectRateDisjointSecond$ for the infection rate of that set is computed (line 18).
Else, members of $\vertexSet_\disjointSetIdxSecond$ are included in set $A$.
Our rationale follows the infection model described in \Cref{sec:infect-model}, 
which implies that the infection probability of the members of an inner set $\disjointSetIdxSecond$ will always be at least equal to the infection probability of the members (of an outer set $\disjointSetIdx$) whose community(ies) are a subset of $\edgeSet_{\vertexSet_\disjointSetIdx}$.
Hence, if an outer set is heavily infected then a corresponding inner set will be heavily infected, too.
In our experiments, we numerically examine the impact of $\thresh$.

Finally, we test all members of set $A$ that are not tested individually (because infection probability is presumably low) using traditional group testing (line 23).


\subsection{Non-adaptive algorithms}\label{sec:nonadapt}



For simplicity, we describe our non-adaptive algorithm using the symmetric case. 

\subsubsection*{Test Matrix}
We
divide the $(\nTests_1+\nTests_2) \times \nItems$  matrix $\testmatrix$ into two sub-matrices $\testmatrix_1$ and $\testmatrix_2$ of sizes $\nTests_1 \times \nItems$ and $\nTests_2 \times \nItems$.

\noindent $\triangleright$ The sub-matrix $\testmatrix_1$ identifies the non-infected outer sets using one mixed sample for each outer set (\Cref{sec:notation:models}).
If the number of tests available is large, we set $\nTests_1$ to be the number of outer sets, i.e., we use one test for each outer set; otherwise, in sparse $\nDefFamilies$ regimes, $\nTests_1$ can be closer to $\mathcal{O}(\nDefFamilies\log\frac{\nFamilies}{\nDefFamilies})$. 

\noindent $\triangleright$ Assume that $\nTests_2=\frac{\nItems}{\nOnesPerRow}$, for some constant $\nOnesPerRow$.
The sub-matrix $\testmatrix_2$ of size $\nTests_2 \times \nItems$ 
has one ``1" in each column (each of the $\nItems$ member participates in one test) and $\nOnesPerRow$  ``1"s in each row (each test pools together $\nOnesPerRow$ members); equivalently, $\testmatrix_2$
is a concatenation of $\nOnesPerRow$ identity matrices $I_{\nTests_2}$, 
i.e., $\testmatrix_2=\left[I_{\nTests_2}~I_{\nTests_2}\cdots I_{\nTests_2}\right]$. For $\nOnesPerRow=1$, this reduces to individual testing.
The design of $\testmatrix_2$ amounts to deciding which members are placed in the same test. We propose that no two members from the same outer set are placed in the same test and that members from the same inner set are placed in the same test ($\nOnesPerRow$ members in each test). 
 
\subsubsection*{Decoding} 
We use the test outcomes of $\testmatrix_1$ to identify the non-infected outer sets and proceed to remove the corresponding columns from $\testmatrix_2$. We next use the remaining columns of  $\testmatrix_2$ 
and combinatorial orthogonal matching pursuit (COMP) to identify the infected members, namely:
$(i)$	A member is identified as non-infected  if it is included in at least one negative test in $\testmatrix_2$. 
$(ii)$	All other members, that are only included in positive tests in $\testmatrix_2$, are identified as infected.

\subsubsection*{Intuition} Suppose infected communities have a large percentage (say $>40\%$) of infected members. The idea is that pooling together multiple highly correlated items in the same test (such as people in the same outer/inner set) enables COMP to mark all these items as non-defective in case of a negative result.

\subsubsection*{Example}
We here illustrate for a special case our proposed design for matrix $\testmatrix_2$ and the resulting error rate our algorithm achieves.
Assume that we have $\nFamilies$ communities, 
where $2\nFamilies_o$  communities  pairwise overlap
(each community overlaps with exactly one other community) 
and the remaining  $\nFamilies-2\nFamilies_o$ communities
do not overlap with any other community.
Assume each community has $\nMembersSymmetric$ members, and overlapping communities share
 $\nMembersSymmetric_o$ members. 
We construct the sub-matrix $\testmatrix_2$  of size $\nTests_2 \times \nItems$ 
as in the following example that uses $\nFamilies=6$, $\nFamilies_o=2$, $\nMembersSymmetric=3$, $\nMembersSymmetric_o=1$:
\begin{figure}[h!]
    \centering
    \vspace{-0.4cm}
    \includegraphics[scale=0.35]{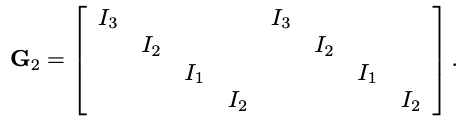}
    \vspace{-0.4cm}
\end{figure}


This matrix starts with $\nBlocks_1=\frac{\nFamilies-2\nFamilies_o}{\nOnesPerRow}$ block-rows that each contains $\nOnesPerRow$ identity matrices  $I_{\nMembersSymmetric}$, one corresponding to each non-overlapping community.
We then have $\nBlocks_2=\frac{\nFamilies_o}{\nOnesPerRow}$ block-rows each containing $\nOnesPerRow$  identity matrices $I_{2\nMembersSymmetric-\nMembersSymmetric_o}$, one for each pair of overlapping communities. Each $I_{2\nMembersSymmetric-\nMembersSymmetric_o}$ matrix contains three matrices $I_{\nMembersSymmetric-\nMembersSymmetric_o}$, $I_{\nMembersSymmetric_o}$, and $I_{\nMembersSymmetric-\nMembersSymmetric_o}$ corresponding to the members that belong only in one of the communities, or in both.
Note that $\nFamilies=(\nBlocks_1+2\nBlocks_2)c$ and $\nTests_2=\nBlocks_1\nMembersSymmetric+\nBlocks_2(2\nMembersSymmetric-\nMembersSymmetric_o)$. 

\subsubsection*{Error Rate} 
Note that our decoding strategy leads to zero FN errors. 
The following lemma provides an analysis of the error (FP) rate for the design of $\testmatrix_2$ in the example which is defined as: 
$R(\text{error})\triangleq \sfrac{1}{\nItems}\cdot|\{\memberIndex:\hat{\defVariable}_\memberIndex\neq \defVariable_\memberIndex\}|$.
We provide the expected error rate for only the probabilistic model (II) for the purpose of comparison with traditional Bernoulli design in Fig.~\ref{fig:joint_figure}(a).

\begin{lemma}
	\label{lemma-ErrorRate-nonadaptive}
	For $\testmatrix_2$ as in the example, the error rate is calculated for the probability model (II) as:
	\begin{align}
	R_{II}(\text{error})&=\frac{1}{\nItems}\Big[\left(1-(1-\memberInfectRateSymmetric\familyInfectRate)^{\nOnesPerRow-1}\right)\cdot N_1 \nonumber \\
	&\qquad\qquad + \left(1-(1-\memberInfectRateSymmetric\familyInfectRate)^{2(\nOnesPerRow-1)}\right)\cdot N_2\Big], \label{ErrorRate-2} 
	\end{align}
	where $N_1$ and $N_2$ are the expected number of non-overlapped and overlapped members in infected communities that are non-infected, respectively, and can be obtained as 
	\begin{align*}
	N_1&=(\nFamilies-2\nFamilies_o)\familyInfectRate(1-\memberInfectRateSymmetric)\nMembersSymmetric + 2\nFamilies_o\familyInfectRate(1-\memberInfectRateSymmetric)(\nMembersSymmetric-\nMembersSymmetric_o)  \\
	N_2&=\nFamilies_o\left(1-(1-\familyInfectRate)^2\right)(1-\memberInfectRateSymmetric)\nMembersSymmetric_o.
	\end{align*}
\end{lemma}

\begin{figure}
    \centering
    \includegraphics[scale = 0.2]{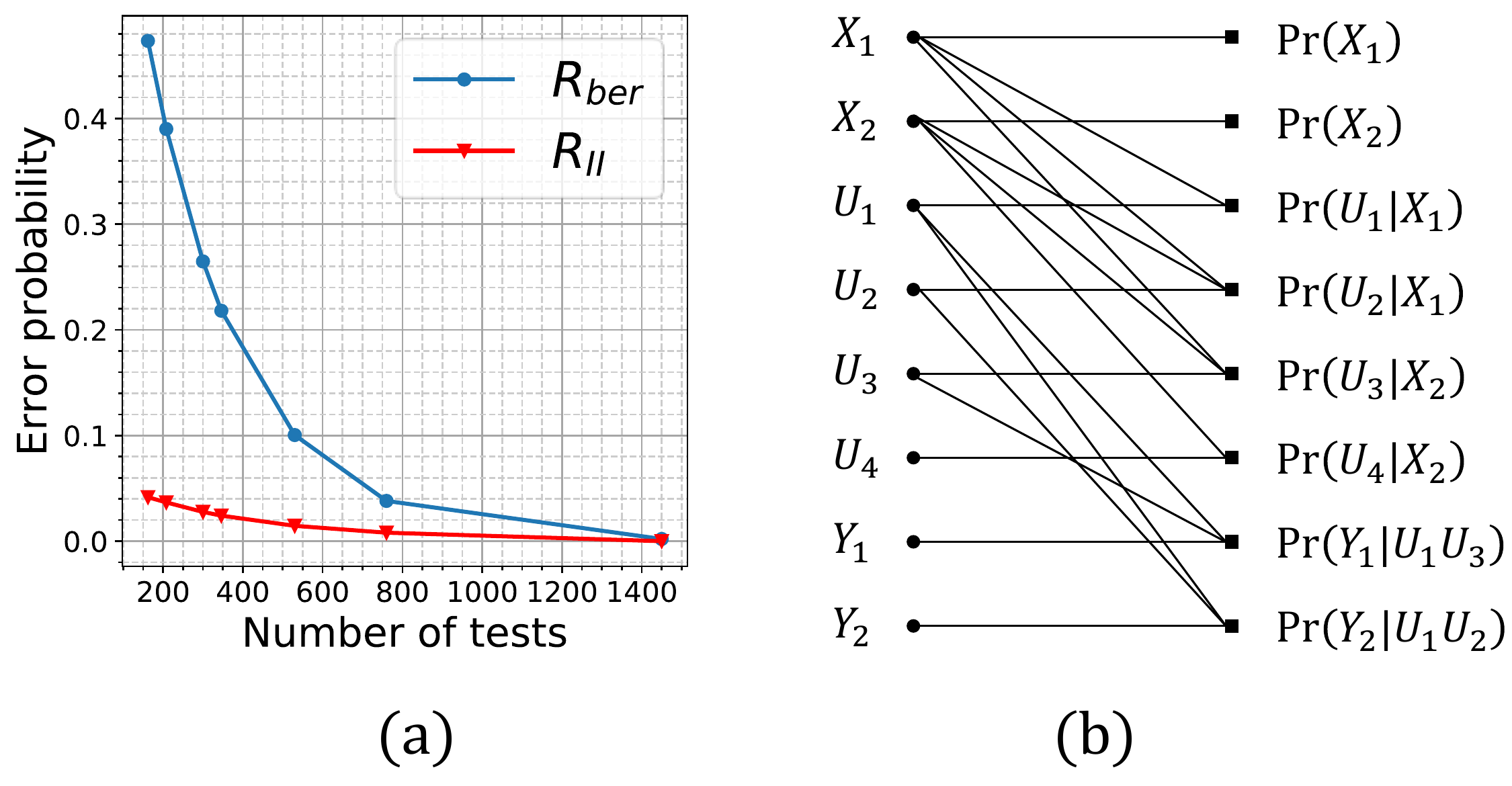}
    \caption{(a). Error rate for Bernoulli design vs $G_1G_2$ design for the example. (b). An example of factor graph.}
    \label{fig:joint_figure}
    \vspace{-0.4cm}
\end{figure}

The error rate of traditional group testing using Bernoulli design (with parameter $\frac{1}{\nDef}$) and COMP decoding has an error rate of 
$R_{\text{tradition}}(\text{error})=\sfrac{1}{\nItems} \cdot (\nItems-\nDef)\left(1-1/\nDef(1-1/\nDef)^k\right)^\nTests.$
Fig.~\ref{fig:joint_figure}(a) depicts $R(\text{error})$ for parameters $\nFamilies=150$, $\nFamilies_o=60$, $\nMembersSymmetric=10$, $\nMembersSymmetric_o=2$, $\familyInfectRate= 0.2$, and $\memberInfectRateSymmetric=0.2$.


\begin{figure*}[t]
\centering
\begin{multicols}{3}
    \includegraphics[width=\linewidth, height = 3.7cm]{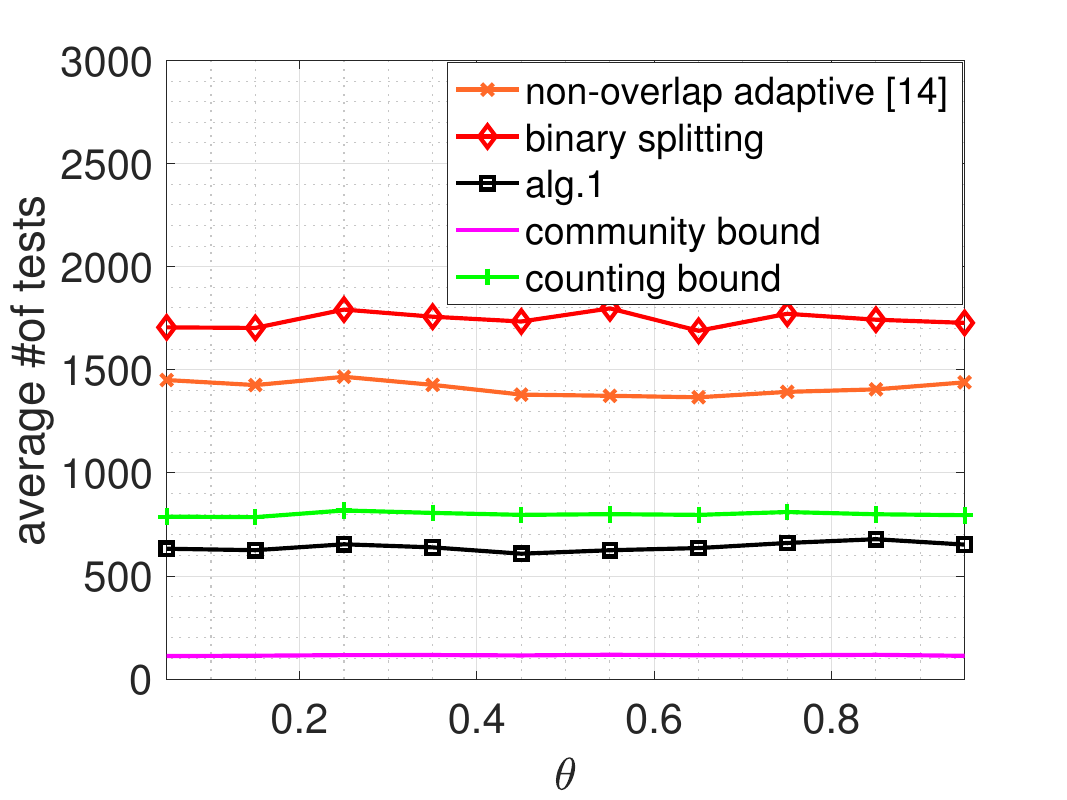}\par	
    \caption{Average number of tests comparison of various adaptive algorithms and combinatorial bound. Here $n=3000$, $F\sim \text{Uniform}[15,25]$, $q=0.05$, $p_e \sim \text{Uniform}[0.3,0.9]$.}
	\label{fig:adaptive}
   \includegraphics[width=\linewidth, height = 3.7cm]{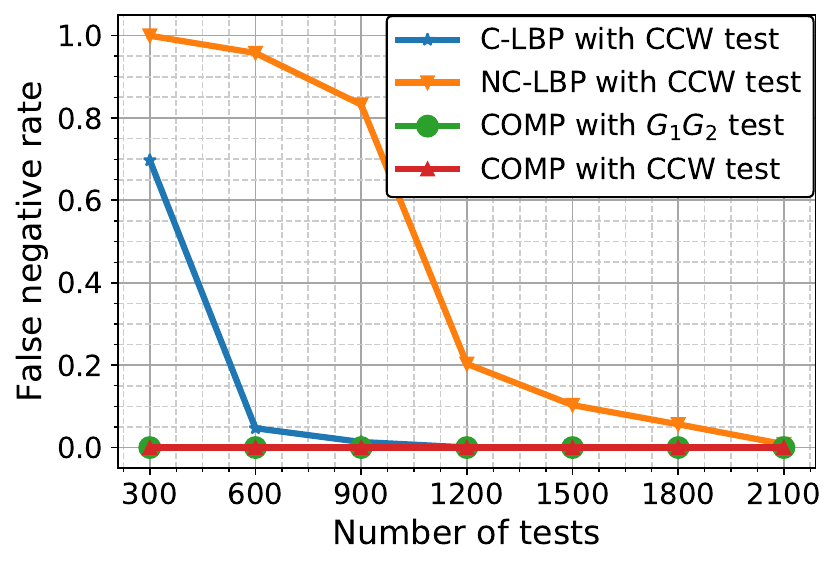}\par	    
   \caption{FN rate comparison of various non-adaptive test designs with corresponding decoding algorithms. Here $n=3000$, $F\sim \text{Uniform}[15,25]$, $q=0.05$, $p_e \sim \text{Uniform}[0.3,0.9]$.}
    \label{fig:fn_nonadaptive}
   \includegraphics[width=\linewidth, height = 3.7cm]{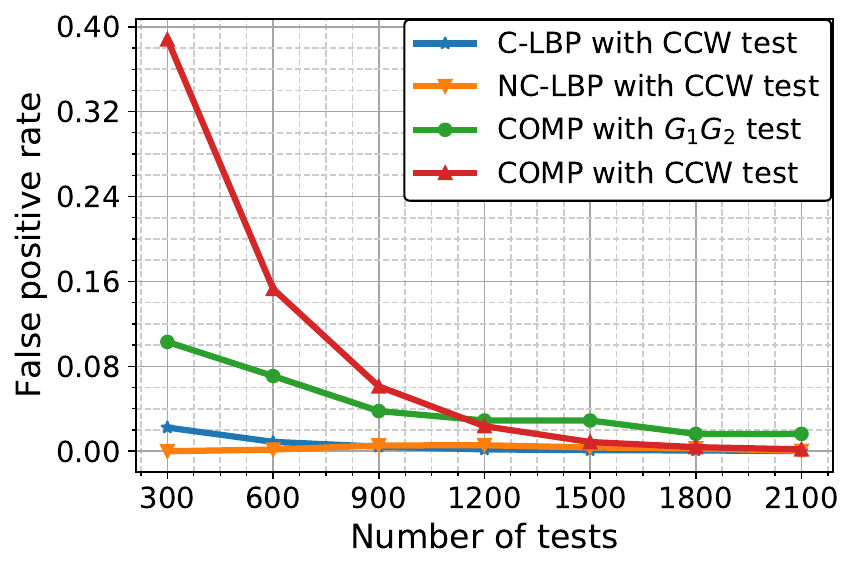}\par	    
   \caption{FP rate comparison of various non-adaptive test designs with corresponding decoding algorithms. Here $n=3000$, $F\sim \text{Uniform}[15,25]$, $q=0.05$, $p_e \sim \text{Uniform}[0.3,0.9]$.}
    \label{fig:fp_nonadaptive}
\end{multicols}
\vspace{-0.8cm}
\end{figure*}

\section{Loopy belief propagation decoder}
\label{sec:LBP}
Apart from COMP, we also use loopy belief propogation (LBP \cite{kschischang2001factor})  to infer the infection status of the individual (and communities). LBP forms an estimate of the posterior probability that an individual (or a community) is infected, given the test results. This estimate is exact when the underlying factor graph describing the joint distribution is a tree, however this is rarely the case. Nevertheless, it is an algorithm of practical importance and has achieved success on a variety of applications.
Also, LBP offers soft information (posterior distributions), which is  more useful than hard decisions in the context of disease-spread management.

We use LBP for our probabilistic model, because it is fast and can be easily configured to take into account the community structure.  
Many inference algorithms exist that estimate the posterior marginals, some of which have also been employed for group testing. 
For example, GAMP~\cite{zhu2020noisy} and Monte-Carlo sampling~\cite{cuturi2020noisy} yield more accurate decoders.
The main focus of this work is to examine whether benefits from accounting for the community structure (both at the test design and the decoder) exist; 
hence we think that considering a simple (possibly sub-optimal) decoder based on LBP is a good first step; 
we defer more complex designs to future work.

We next describe the factor graph and the belief propagation update rules for our probabilistic model (II).
Let the infection status of each community $\familyIndex$ be $\defFamilyVariable_\familyIndex \sim \Ber(\familyInfectRate)$. 
Moreover, let $S_v$ denote the set of communities that $\defVariable_{\memberIndex}$ belongs to. Then:
\begin{align}\label{eq:factorizedDistribution}
\Pr(\defFamilyVariable_1,&...,\defFamilyVariable_\nFamilies,\defVariable_{1},...,\defVariable_{\nItems},Y_{1},...,Y_{\nTests}) = \nonumber\\
&\prod_{\familyIndex=1}^\nFamilies \Pr(\defFamilyVariable_\familyIndex) \prod_{\memberIndex=1}^\nItems \Pr(\defVariable_{\memberIndex}|\defFamilyVariable_{S_v}) \prod_{\testIndex = 1}^\nTests \Pr(Y_{\testIndex}|\defVariable_{\delta_{\testIndex}}),
\end{align}
where $\delta_{\testIndex}$ is the group of people included in test $\testIndex$.
\Cref{eq:factorizedDistribution} can be represented by a factor graph, where the variable nodes correspond to the variables  $\defFamilyVariable_\familyIndex,\defVariable_{\memberIndex},Y_{\testIndex}$ and the factor nodes correspond to $\Pr(\defFamilyVariable_\familyIndex), \Pr(\defVariable_{\memberIndex}|\defFamilyVariable_{S_v}), \Pr(Y_{\testIndex}|\defVariable_{\delta_{\testIndex}})$; 
Fig.~\ref{fig:joint_figure}(b) shows an example of $2$ communities, $4$ members and $2$ tests.

Given the result of each test is $y_\testIndex$, i.e., $Y_{\testIndex}=y_{\testIndex}$, 
LBP estimates the marginals $\Pr(\defFamilyVariable_\familyIndex=v|Y_1=y_1,...,Y_\nTests=y_T)$  and  $\Pr(\defVariable_{\memberIndex}=u|Y_1=y_1,...,Y_\nTests=y_\nTests)$, 
by iteratively exchanging messages across the variable and factor nodes.
The messages are viewed as \textit{beliefs} about that variable or distributions (a local estimate of $\Pr(\text{variable}|\text{observations})$). 
Since all random variables are binary, each message is a 2-dimensional vector. 

We use the factor graph framework from \cite{kschischang2001factor} to compute the messages:
Variable nodes $Y_\testIndex$ continually transmit the message $[0,1]$ if $y_\testIndex=1$ and $[1,0]$ if $y_\testIndex = 0$ on its incident edge, at every iteration.
Each other variable node ($\defFamilyVariable_\familyIndex$ and $\defVariable_{\memberIndex}$) uses the following rule: 
for incident edge $\epsilon$, the node computes the elementwise product of the messages from every other incident edge   and transmits this along $\epsilon$.
For the factor node messages, we derive closed-form expressions for the sum-product update rules (akin to equation (6) in \cite{kschischang2001factor}). 
\pavlos{The exact messages are described in the Appendix of our extended version \cite{techRpt}.}


\section{Numerical evaluation}\label{section-experiments} 

In this section, we evaluate the benefit of accounting for the community structure, in terms of error rate and number of tests required, using $100$ random structures, each having $\nItems=3000$ members participating in about $200$ overlapping communities.

\noindent\textbf{Experimental setup.}
We generate each structure using the following rules:
the size of each community is selected uniformly at random from the range $[15,25]$,
and
each member is randomly allocated in at most $4$ communities (according to a geometric distribution). 
Then, the members become infected according to the probabilistic model (II): 
each community $\edgeIndex$ gets infected w.p. $\familyInfectRate=0.05$; 
and if infected, then its infection rate $\memberInfectRate$ is randomly chosen from the interval $[0.3,0.9]$.
We remark that our experimental setup yields a linear infection regime; the fraction of infected members about $5\%$ overall. 
We preferred such a setup in order to stress the performance of our algorithms, as we know that group testing generally shows less benefits in linear regimes. 

For the adaptive algorithms, we compare:  
the binary splitting algorithm (BSA)~\cite{GroupTestingMonograph},
which is the best traditional alternative when 
the number of infected members is unknown;
the algorithm proposed in~\cite{GroupTesting-community} that considers communities but no overlap;
and Alg.~\ref{algorithm-noiseless} with BSA in the place of $\adapt()$.

For the non-adaptive test matrix designs, we compare:
\textit{$G_1G_2$},  our proposed test design  in \Cref{sec:nonadapt};
and \textit{CCW}, constant-column-weight algorithms, where each item is included in a fixed number $w$ of tests selected uniformly at random. $w$ is assumed to be of the form $w=\alpha \frac{T}{k}$, where $k$ is an estimate of the number of defectives in the population. We  
exhaustively search to find the best value of $\alpha\in [0,1]$.

We also compare LBP and COMP decoding:
\textit{C-LBP} is our proposed algorithm in section \ref{sec:LBP}, that  takes into account the community structure.
\textit{NC-LBP}, does not take into account the community structure, i.e., assumes that each individual is i.i.d infected with the same probability $p_{iid}$.
\textit{COMP}, described in~\cite{GroupTestingMonograph}, has a zero FN probability by design.

\noindent\textbf{Results.} \\
$(i)$ \textit{Adaptive test designs.} 
\pavlos{For each community structure, 
we measured the number of tests needed by each algorithm to achieve zero-error identification.
Since Alg.~\ref{algorithm-noiseless} depends on $\thresh$, the threshold used at line $16$, we evaluated its performance for various values of $\thresh$.
\Cref{fig:adaptive} depicts the average performance of our algorithm (for each $\thresh$, we average over $100$ randomly generated structures). 
Alg.~\ref{algorithm-noiseless} was proved resilient to the choice of $\thresh$ and needed on average $>55\%$ fewer tests than the other algorithms.
Its performance was also better than the counting bound, 
which is our best hope with traditional group testing. 
Our findings were similar for sparser infection regimes (see results in extended version~\cite{techRpt}), and there were cases where our algorithm performed closer to the community bound~\eqref{eq:combinatorialBound}.}

\noindent $(ii)$ \textit{Non adaptive test designs}. In our experiments, we measured the FN/FP rates achieved by the non-adaptive test designs and the corresponding decoders. Fig.~\ref{fig:fn_nonadaptive} and Fig.~\ref{fig:fp_nonadaptive} depict FN and FP rates as a function of $\nTests\in[300,2100]$, respectively. 
The key takeaways are as follows: \\
$\bullet$ C-LBP with CCW attains zero FP and FN at $1200$ tests while  COMP and NC-LBP with CCW (which are agnostic to the community structure) attain zero FP and FN only at $1800$ and $2100$ tests respectively. This illustrates potential benefits of making the decoder aware of the community structure.
\\ $\bullet$ If we desire a zero FN rate (or if we would like to use a simple decoder) and we are constrained to use less than $1000$ tests, the $G_1G_2$ test design with COMP gives the lowest FP rates. This illustrates the benefit of designing tests matrices that take into account the community structure.


\section{Conclusion}
\label{sec:conclusion}
\pavlos{
This work attempted to relax the ``independent-infections'' assumption (that is traditionally made) in group testing, by introducing correlated infections through an overlapping community structure.
In our proposed infection model, an individual belongs to one or more communities, and the infection probability depends on the infected communities (s)he participates in---a model that is suitable for universities or schools, where students join various classrooms.
Given this model, we derived new lower bounds for the number of tests needed in order to perfectly identify the status of each individual,
and we provided adaptive and non-adaptive testing algorithms that incorporate the community structure in their test designs or decoding. 
Our algorithms are not always optimal w.r.t.\ the lower bounds, but perform significantly better than community-agnostic group testing; 
per our experiments, they need $30\%$-$65\%$ fewer tests (on average) to achieve the same identification accuracy.
In our opinion, this result illustrates the important savings that community structure may offer and shows that it is worth investigating more sophisticated models and/or algorithms in the future.}

\bibliographystyle{IEEEtran}
\bibliography{refs}

\appendix
\section{Appendix for Section~\ref{sec:lower-bounds}: lower bounds}
\label{app:sec:lower-bounds}
    \subsection{Proof of Theorem~\ref{thm:combinatorialBound}}
\begin{proof}
	\pavlos{There exist
	$\binom{\nFamilies}{\nDefFamilies} \cdot \prod_{\component \in \graph} \prod_{\disjointSetIdx \in \disjointSets} \binom{|\vertexSet_\disjointSetIdx|}{\nDefMembersDisjointSet}$ possible combinations of infected members.
 	This is because 
	there are $\binom{\nFamilies}{\nDefFamilies}$  combinations of infected communities, each of which has the same chance of occurring,
	and is associated with a structure of connected components. 
	In each disjoint set $\disjointSetIdx \in \disjointSets$ of every connected component $\component \in \graph$, there are $\binom{|\vertexSet_\disjointSetIdx|}{\nDefMembersDisjointSet}$ possible combinations of infected members, each of which has the same chance of occurring\footnote{Note that the product is over all disjoint sets instead of only the infected ones, because $\binom{|\vertexSet_\disjointSetIdx|}{\nDefMembersDisjointSet} = 1$, whenever $\nDefMembersDisjointSet = 0$.}.
	
	To achieve zero-error identification, each combination of infected members must give a different set of test results. 
	Given that there are only $2^\nTests$ possible results, we need: $2^\nTests \ge \binom{\nFamilies}{\nDefFamilies} \cdot \prod_{\component \in \graph} \prod_{\disjointSetIdx \in \disjointSets} \binom{|\vertexSet_\disjointSetIdx|}{\nDefMembersDisjointSet},$ which completes the proof.}
\end{proof}

\subsection{Proof of Theorem~\ref{thm:probabilisticBound}}
\begin{proof}
\snote{
Let 
$\defFamilyVariables$ be the indicator random vector for the infection status of all communities.
By rephrasing~\cite[Theorem 1]{prior}, any probabilistic group testing algorithm using $\nTests$ noiseless tests can achieve a zero-error reconstruction of $\defVector$ if:
	\begin{equation}
	\nTests \ge \entropy(\defVector) = \entropy(\defFamilyVariables) + \entropy(\defVector|\defFamilyVariables) - \entropy(\defFamilyVariables|\defVector).
	\end{equation} 	
	The first term is:
	    $\entropy({\defFamilyVariables}) = \sum_{\familyIndex=1}^\nFamilies \entropy(\defFamilyVariable_\familyIndex) = \nFamilies \binEntropy(\familyInfectRate).$
	
	\noindent The second term is calculated as:
	\begin{align*}
	    &\entropy({\bf U}|{\defFamilyVariables}) \overset{(a)}{=} \sum_{\vertexIndex=1}^n \entropy(U_v|\defFamilyVariables_{\edgeSet_\vertexIndex}) \\
	     &=\sum_{\vertexIndex=1}^n \sum_{{\bf x} \in \{0,1\}^{|\edgeSet_\vertexIndex|}}\Pr({\defFamilyVariables}_{\edgeSet_\vertexIndex}={\bf x}) \entropy(U_v|\defFamilyVariables_{\edgeSet_\vertexIndex}=\bf x)\\
	    &\overset{(b)}{=} \sum_{\vertexIndex=1}^n \sum_{\edgeSubset \subseteq \edgeSet_\vertexIndex} \familyInfectRate^{|\edgeSubset|}(1-\familyInfectRate)^{|\edgeSet_\vertexIndex|-|\edgeSubset|} \entropy(U_v|\defFamilyVariables_{\edgeSubset}=\textbf 1,\defFamilyVariables_{\edgeSet_\vertexIndex\setminus \edgeSubset}=\mathbf 0)\\
	    &= \sum_{\vertexIndex=1}^n \sum_{\edgeSubset \subseteq \edgeSet_\vertexIndex} \familyInfectRate^{|\edgeSubset|}(1-\familyInfectRate)^{|\edgeSet_\vertexIndex|-|\edgeSubset|}  \binEntropy(\prod_{\familyIndex \in \edgeSubset}(1-\memberInfectRate)),
	\end{align*}
	where in $(a)$, $\edgeSet_\vertexIndex$ refers to the set of communities member $\vertexIndex$ belongs to, and in $(b)$ the subset $\edgeSubset$ is the subset of infected communities in $\edgeSet_\vertexIndex$.
	
	\noindent Finally, we upper bound the third term as:
	\begin{align*}
	    \entropy&(\defFamilyVariables|\defVector) \leq \sum_{\familyIndex=1}^\nFamilies \entropy(\defFamilyVariable_\familyIndex|\defVector)
	    = \sum_{\familyIndex=1}^\nFamilies \entropy(\defFamilyVariable_\familyIndex|\defVector_{\mathcal S_e})\\
	    &= \sum_{\familyIndex=1}^\nFamilies \Pr(\defVector_{S_e}=\mathbf 0) \binEntropy(\Pr(\defFamilyVariable_\familyIndex=0|\defVector_{S_e}=\mathbf 0))\\
	    &= \sum_{\familyIndex=1}^\nFamilies (1-\familyInfectRate+\familyInfectRate(1-\memberInfectRate)^{|S_e|}) \binEntropy\left (\frac{1-\familyInfectRate}{1-\familyInfectRate+\familyInfectRate(1-\memberInfectRate)^{|S_e|}}\right ),
	\end{align*}
	where $S_e$ is the set of members who \textit{only} belong to community $\familyIndex$. 
	Combining all the 3 terms
 	concludes the proof.}
\end{proof}


\section{Appendix for \cref{sec:nonadapt}: The Noiseless Nonadaptive case}
\label{app:sec:Non-adaptive}

\subsection{Proof of \Cref{lemma-ErrorRate-nonadaptive}}
\label{app:lemma-ErrorRate-nonadaptive-proof}
For a non-overlapped non-infected member $\memberIndex$ that belongs to only one community, the probability that $\memberIndex$ is misidentified as infected is $1-(1-\memberInfectRateSymmetric\familyInfectRate)^{\nOnesPerRow-1}$. 
For an overlapped non-infected member $\memberIndex$ that belongs to more than one communities, the probability that $\memberIndex$ is misidentified as infected is $1-(1-\memberInfectRateSymmetric\familyInfectRate)^{2(\nOnesPerRow-1)}$. 
Note that we assume the decoding of $\testmatrix_1$ has no errors, i.e., it identifies all non-infected outer sets correctly. 
Then for the pairwise overlap structure in the example, the infection status of all non-overlapped communities and non-overlapped parts are identified correctly. 
The COMP decoding of $\testmatrix_2$ has no FNs. 
The expected total number of FPs $N_0$ can be obtained as $N_0\leq \left(1-(1-\memberInfectRateSymmetric\familyInfectRate)^{\nOnesPerRow-1}\right)\cdot N_1 + \left(1-(1-\memberInfectRateSymmetric\familyInfectRate)^{2(\nOnesPerRow-1)}\right)\cdot N_2$, where the inequality is because the RHS have not used the testing resluts of $\testmatrix_1$, $N_1$ and $N_2$ are the expected number of non-overlapped and overlapped members in infected communities, respectively. 
We can calculate $N_1$ as follows,
\begin{align}
N_1=(\nFamilies-2\nFamilies_o)\familyInfectRate(1-\memberInfectRateSymmetric)\nMembersSymmetric + 2\nFamilies_o\familyInfectRate(1-\memberInfectRateSymmetric)(\nMembersSymmetric-\nMembersSymmetric_o),
\end{align}
where $(\nFamilies-2\nFamilies_o)\familyInfectRate$ is the expected number of infected non-overlapped communities, 
$(1-\memberInfectRateSymmetric)\nMembersSymmetric$ is the expected number of non-infected members in each infected non-overlapped community, 
$2\nFamilies_o\familyInfectRate$ is the expected number of infected overlapped communities, 
and $(1-\memberInfectRateSymmetric)(\nMembersSymmetric-\nMembersSymmetric_o)$ is the expected number of non-infected members in each infected overlapped community. 
Similarly, $N_2$ can be calculated as 
\begin{align}
N_2=\nFamilies_o\left(1-(1-\familyInfectRate)^2\right)(1-\memberInfectRateSymmetric)\nMembersSymmetric_o,
\end{align}
where $\nFamilies_o\left(1-(1-\familyInfectRate)^2\right)$ is the expected number of overlaps, 
and $(1-\memberInfectRateSymmetric)\nMembersSymmetric_o$ is the expected number of non-infected members in each overlapped part.

\end{document}